\documentclass[a4paper,12pt]{amsart}

\usepackage[utf8]{inputenc}
\usepackage[T1]{fontenc}
\usepackage{enumerate}
\usepackage{amsmath, amsthm, amssymb}
\usepackage{mathtools}
\usepackage{standalone}
\usepackage[url=false,isbn=false,giveninits=true]{biblatex}
\newtheorem{proposition}{Proposition}
\newtheorem{corollary}{Corollary}
\newtheorem{theorem}{Theorem}
\newtheorem{lemma}{Lemma}

\theoremstyle{definition}
\newtheorem{definition}{Definition}
\newtheorem{example}{Example}

\newtheorem{remark}{Remark}

\DeclarePairedDelimiter\inner{\langle}{\rangle}
\DeclarePairedDelimiter\abs{\lvert}{\rvert}
\DeclarePairedDelimiter\norm{\lVert}{\rVert}

\DeclareMathOperator{\vspan}{Span}
\DeclareMathOperator{\Tr}{Tr}
\DeclareMathOperator{\Sp}{Sp}

\newcommand{\F}{\mathbb{F}}
\newcommand{\C}{\mathbb{C}}
\newcommand{\R}{\mathbb{R}}
\newcommand{\Z}{\mathbb{Z}}
\newcommand{\N}{\mathbb{N}}

\begin{document}

\title[Rigidity and a common framework for MUBs and $k$-nets]{Rigidity and a common framework for \\ mutually unbiased bases and $k$-nets}

\author{Sloan Nietert}
\address{
Sloan Nietert\\
Dept.\ of Computer Science,
Cornell University\\
Ithaca, NY 14853, USA}
\email{sbn45@cornell.edu}

\author{Zsombor Szil\'agyi}
\address{
Zsombor Szil\'agyi\\
MTA-BME Lend\"ulet Quantum Information Theory Research Group\\
Dept.\ of Theoretical Physics, 
Budapest University of Technology \& Economics (BME)\\
Budafoki \'ut 8., H-1111 Budapest, Hungary
}
\email{zsombor.szilagyi@gmail.com}

\author{Mih\'aly Weiner}
\address{
Mih\'aly Weiner\\
MTA-BME Lend\"ulet Quantum Information Theory Research Group\\
Dept.\ of Analysis, 
Budapest University of Technology \& Economics (BME)\\
Egry J\'ozsef u. 1., H-1111 Budapest, Hungary
}
\email{mweiner@math.bme.hu}

\thanks{Sz. Zsombor and M. Weiner are supported in part by grant NRDI K 124152 and  KH 129601. M. Weiner is also supported by the ERC advanced grant 669240 QUEST ``Quantum Algebraic Structures and Models.'' S. Nietert was supported by the Hungarian-American Fulbright Commission.}

\begin{abstract}
Many deep, mysterious connections have been observed between
collections of mutually unbiased bases (MUBs) and combinatorial designs called $k$-nets (and in particular, between \textit{complete} collections of MUBs and finite affine --- or equivalently: finite projective --- planes). Here we introduce the notion of a $k$-net over an algebra $\mathfrak{A}$ and thus provide a common framework for both objects. In the commutative case, we recover (classical) $k$-nets, while choosing  $\mathfrak{A} := M_d(\C)$ leads to collections of MUBs. 

A common framework allows one to find shared properties
and proofs that ``inherently work'' for both objects. As a first example, we derive a certain rigidity property which was previously shown to hold for $k$-nets that can be completed to affine planes using a completely different, combinatorial argument. For $k$-nets that cannot be completed and for MUBs, this result is new, and, in particular, it implies that the only vectors unbiased to all but $k \leq \sqrt{d}$ bases of a complete collection of MUBs in $\C^d$ are the elements of the remaining $k$ bases (up to phase factors). In general, this is false when $k$ is just the next integer after $\sqrt{d}$; we present an example of this in every prime-square dimension, demonstrating that the derived bound is tight.

As an application of the rigidity result, we prove that if a large enough collection of MUBs constructed from a certain type of group representation (e.g.\ a construction relying on discrete Weyl operators or generalized Pauli matrices) can be extended to a complete system, then in fact \textit{every} basis of the completion must come from the same representation. In turn, we use this to show that certain large systems of MUBs cannot be completed.
\end{abstract}

\maketitle

\setlength{\parindent}{0em}
\setlength{\parskip}{\medskipamount}

\section{Introduction}\label{sec:intro}

Mutually unbiased bases (MUBs) arise naturally in several quantum information protocols and are investigated extensively from both purely mathematical and quantum informational perspectives
\cite{schwinger60,ivanovic81,wootters89, weiner13, self}.
Recall that two orthonormal bases $\mathcal{E}$ and $\mathcal{F}$ of $\C^d$ are called \textit{mutually unbiased} if $\abs{\inner{\mathbf{e},\mathbf{f}}}^2 = 1/d$ for each $\mathbf{e} \in \mathcal{E}, \mathbf{f} \in \mathcal{F}$ and that a collection of pairwise mutually unbiased
bases $\mathcal{E}_1, \dots, \mathcal{E}_r$
in $\C^d$ is said to be \textit{complete}, if $r=d+1$ (as it is easy to prove that $r \leq d+1$ for any collection of MUBs in $\C^d)$.
Researchers widely believe that a complete set of MUBs exist in $\C^d$ if and only if $d$ is a prime power, but non-existence has yet to be proven for even a single dimension. There are, however, examples of collections of MUBs which cannot be completed, even in dimensions for which complete systems exist; see e.g.\ \cite{szanto16} and the references therein.

\subsection{Combinatorial $k$-nets}

Many have observed that complete collections of MUBs closely resemble finite affine and projective planes, two combinatorial designs from finite geometry easily shown to be equivalent; see \cite{bengtsson17} for a good overview and note the comparison drawn in \cite{wootters06} between orthogonal projections and lines, which we shall use to provide a common framework for our objects of interest. In particular, these designs are known to exist for all prime-power orders and are conjectured to exist for no others. However, even when these structures, which are in some sense ``complete,'' do not exist, we have related ``incomplete'' structures called \textit{$k$-nets} (or equivalently: incomplete collections of mutually orthogonal Latin squares) that are closely tied to (incomplete) collections of MUBs; see e.g.\ the nice construction of \cite{bw05}.

\begin{definition}
A \textit{$k$-net} is an incidence structure consisting of a set $X$ (whose elements are called \textit{points}) and a collection of subsets of $X$ (called \textit{lines}) such that
\begin{enumerate}[(i)]
    \item the relation $||$ --- where $\ell_1 ||\,\ell_2$ means that $\ell_1 = \ell_2$ or $\ell_1 \cap \ell_2 = \emptyset$ --- is an equivalence relation dividing the set of lines into $k$ equivalence classes (called \textit{parallel classes});
    \item any two lines are either parallel or intersect at a single point;
    \item for any point $p$ and line $\ell$, $\exists$ a line parallel to $\ell$ containing $p$.
\end{enumerate}
\end{definition}

Note that this last property is equivalent to the union of each parallel class being all of $X$. Hence, for $k \geq 2$, every line of a parallel class must contain as many points as there are lines in any other class, and, if $k \geq 3$, this number $d$ --- called the \textit{order} of the $k$-net --- must be the same for all parallel classes. Thus, a \textit{$k$-net of order $d$} consists of $d^2$ points and $k$ parallel classes such that each class has exactly $d$ lines and each line has exactly $d$ points. A ($d+1$)-net of order $d$ is called an \textit{affine plane of order $d$}. Simple arguments show $k \leq d+1$ for all $k$-nets of order $d$, so one might say that affine planes are \textit{complete} $k$-nets.

Now, suppose that we have a ($d+1-k$)-net of order $d$ which we would like to complete to an affine plane. This may not be possible, but, if such a completion exists, we can ask whether it is unique. In \cite{bruck}, Bruck proved that if $k < \sqrt{d}+1$, then a ($d+1-k$)-net of order $d$ has at most $k d$ \textit{transversals}, i.e.\ sets having a single point intersection with each line of the net. Thus, if $k < \sqrt{d}+1$ and our net can be completed to an affine plane, each
transversal must be a line of the completion, i.e.\ we have uniqueness. So affine planes have a certain \textit{rigidity}; they are determined by a proper subset of their parallel classes and cannot be ``slightly modified'' while maintaining their defining properties.

Recall that for each prime power $q = p^\alpha$, there exists a finite field $\F_q$. $\F_q^2$ is naturally an affine plane of order $q$, with lines being subsets of the form $\{ (a,b)t + (x,y) \mid t \in \F_q \}$ where $a,b,x,y \in \F_q, (a,b) \neq 0$. If $q=p^2$ is the square of a prime, then $\Z_p^2$ (considered as a subset of $\F^2_{p^2}$) is \textit{not} a line of our plane but still intersects each line of $p^2-p$ parallel classes at a single point. Indeed, it is easy to check that for $a \in \F_{p^2} \setminus \Z_p$, the intersection
\begin{equation*}
    \Z_p^2 \cap \{ (a,1)t + (x,y) \mid t \in \F_{p^2} \}
\end{equation*}
cannot contain multiple points, and thus, by a simple counting argument, must contain exactly one point. Furthermore, all cosets of $\Z_p^2$ share this property, so we have an entire ``fake parallel class.'' Therefore, the bound of $k < \sqrt{d}+1$ required for rigidity is not only sufficient, but also necessary in general, i.e.\ it is sharp.

\subsection{Mutually unbiased bases}

Let us now return to MUBs. Since the mutually unbiased relation depends only on the one-dimensional subspaces spanned by basis vectors, we will regard two orthonormal bases of $\C^d$ as equivalent if they give the same coordinate axes, with basis vectors differing only by complex phase factors. 

As previously mentioned, it is well-known that if $\mathcal{E}_1, \dots, \mathcal{E}_r$ are MUBs in $\C^d$, then $r \leq d + 1$. Now take the collection to be complete, with $r = d+1$, and consider a basis vector $\mathbf{b}$ belonging to one of the last $k$ bases. Clearly, $\mathbf{b}$ is unbiased with respect to the first $d + 1 - k$ bases, i.e.\ $\abs{\inner{\mathbf{e}, \mathbf{b}}}^2 = 1/d$ for each $\mathbf{e} \in \mathcal{E}_1, \dots, \mathcal{E}_{d+1-k}$. One might wonder how small $k$ must be (with respect to $d$) for the \textit{reverse} to hold; that is, for the elements of the remaining $k$ bases to be the \textit{only} unit vectors (up to phase factors) unbiased with respect to $\mathcal{E}_1, \dots, \mathcal{E}_{d+1-k}$.

In Section \ref{sec:niceMUBs}, mirroring the example given for affine planes of order $p^2$, we shall construct in each prime-square dimension $d = p^2$ a complete collection of MUBs $\mathcal{F}_1, \dots, \mathcal{F}_{p^2+1}$ and an entire orthonormal basis of vectors which are unbiased with respect to $\mathcal{F}_1, \dots, \mathcal{F}_{p^2-p}$ but do \textit{not} belong to the final $p + 1$ bases (even accounting for phase factors). This demonstrates that an analogous notion of rigidity for MUBs also fails for $k \geq \sqrt{d}+1$. The similarities do not end here: in Section \ref{sec:rigidity}, we show that for $k \leq \sqrt{d}$, any unit vector unbiased with respect to $\mathcal{E}_1, \dots, \mathcal{E}_{d+1-k}$ belongs to the remaining $k$ bases (up to phase factors).

Actually, we prove an even stronger, more general theorem that also applies to both collections of MUBs which cannot be completed and classical $k$-nets. To understand this result, one should view MUBs as \textit{quasi-orthogonal maximal abelian *-subalgebras} (MASAs) of $M_d(\C)$. Two orthonormal bases $\mathcal{E}$ and $\mathcal{F}$ are mutually unbiased if and only if the corresponding MASAs $\mathcal{A}_\mathcal{E}$ and $\mathcal{A}_\mathcal{F}$ are \textit{quasi-orthogonal}, as we explain in the next section. Systems of quasi-orthogonal *-subalgebras of $M_d(\C)$ are studied both in general \cite{petz,pszw,weiner10} and in particular for their applications to MUBs \cite{weiner13,krusekamp14,szanto16}.

Quasi-orthogonal MASAs corresponding to a collection $\mathcal{E}_1, \dots, \mathcal{E}_{r}$ of MUBs in $\C^d$ span $M_d(\C)$ if and only if $r = d+1$, i.e.\ when the collection is complete. In this case, for each $j$,
\begin{equation*}
    \mathcal{A}_{\mathcal{E}_1} + \dots + \mathcal{A}_{\mathcal{E}_j} = \C I + \left(\mathcal{A}_{\mathcal{E}_{j+1}} + \dots + \mathcal{A}_{\mathcal{E}_{d+1}}\right)^\perp,
\end{equation*}
and so a unit vector $\mathbf{b}$ is unbiased with respect to $\mathcal{E}_{k+1}, \dots, \mathcal{E}_{d+1}$ if and only if the orthogonal projection onto the one-dimensional subspace $\C \mathbf{b}$ lies in the subspace sum $\sum_{j=1}^k \mathcal{A}_{\mathcal{E}_j}$. For this reason, rather than proving the desired rigidity property for complete systems, we will show the following, stronger statement: if $\mathcal{A}_1, \dots, \mathcal{A}_k$ are quasi-orthogonal MASAs in $\C^d$ with $k \leq \sqrt{d}$, then the only rank-one orthogonal projections in $\sum_{j=1}^k \mathcal{A}_{\mathcal{E}_j}$ are those of the listed subalgebras.

\subsection{Generalized rigidity and applications}

Interestingly, our argument does not fully require that the objects of interest lie in $M_d(\C)$.  In the next section, we develop a common framework for MUBs and classical $k$-nets, introducing the notion of a \textit{$k$-net over an algebra $\mathfrak{A}$}, where $\mathfrak{A}$ is a finite-dimensional $C^*$-algebra 
%considered
with its normalized trace $\tau$. 
These generalized nets are comprised of certain orthogonal projections, which correspond to the rank one projections of a system of quasi-orthogonal MASAs when $\mathfrak{A} = M_d(\C)$. When $\mathfrak{A}$ is the commutative algebra of complex functions on a finite set $X$, then these projections are the indicator functions of the lines of a classical $k$-net on $X$.

In Section \ref{sec:rigidity}, we shall prove our main rigidity theorem: if a collection of orthogonal projections ${\mathcal{N} \subset \mathfrak{A}}$ forms a $k$-net of order $d$ over $\mathfrak{A}$, with $k \leq \sqrt{d}$, then any orthogonal projection ${P \in \vspan(\mathcal{N})}$ with $\tau(P) = 1/d$ must belong to $\mathcal{N}$. In case of a combinatorial $k$-net, this is a dual version of the cited theorem of Bruck. Indeed, consider a $k$-net $\mathcal{S}$ of order $d$ that can be completed to an affine plane $\mathcal{P}$, and let $\mathcal{S}' = \mathcal{P} \setminus \mathcal{S}$ be the ($d+1-k$)-net formed by the parallel classes not in $\mathcal{S}$. One has that $\ell$ is a transversal of $\mathcal{S}'$ if and only if it consists of $d$ points and its indicator function is a linear combination of indicator functions of the lines of $\mathcal{S}$. Thus in this case, our theorem and that of Bruck are essentially\footnote{Since $k \in \N$, the inequalities $k \leq \sqrt{d}$ (our bound) and $k < \sqrt{d}+1$ (Bruck's bound) coincide when $d$ is a square, but, in general, the latter is better by $1$. We believe that by some elementary but rather cumbersome computation ---
see Remark \ref{betterboundremark} ---
we could have tightened our bound and reproduced that of Bruck. However, when $d = p^2$, our bound is already tight, while,
in general, it is likely that even Bruck's bound is suboptimal.
For this reason, we decided against pursuing the issue further.} the same, with ours being stated in terms of the $k$-net $\mathcal{S}$ and Brook's in terms of the ($d+1-k$)-net $\mathcal{S}'$. However, even in the commutative classical case --- since in general a $k$-net cannot be completed to an affine plane --- neither our theorem nor that of Bruck seems to directly imply the other.

We find it quite interesting in itself that a common framework can uncover such non-trivial properties of very different mathematical structures. However, our work would be incomplete without discussing some applications of the rigidity result. Hence, in Section \ref{sec:niceMUBs}, we consider systems of MUBs constructed from certain unitary projective group representations, so-called \textit{nice error bases} with abelian index groups (e.g.\ constructions relying on discrete Weyl operators or generalized Pauli matrices). Many of the known complete systems of MUBs are given in this form, so such constructions certainly warrant closer analysis.

Suppose that one uses a nice error basis with abelian index group to construct a collection of $d + 1 - k$ MUBs, with $k \leq \sqrt{d}$. Exploiting the symmetries given by the underlying group representation, we show that the unique completion of this system (if it exists) must be constructed from the same nice error basis. Thus, when $k \leq \sqrt{d}$, the weak statement ``this collection cannot be completed \textit{using our nice error basis}'', cf. \cite{mandayam14}, becomes the following stronger one: ``this collection cannot be completed.'' In particular, we discuss the example treated in \cite{szanto16}. There, in every prime-square dimension $d = p^2$ with $p \equiv 3 \pmod 4$, a complete collection of MUBs $\mathcal{E}_1, \dots, \mathcal{E}_{p^2+1}$ and two additional bases $\mathcal{B}_1, \mathcal{B}_2$ are given such that $\mathcal{E}_1, \dots, \mathcal{E}_{p^2-p}$ together with $\mathcal{B}_1, \mathcal{B}_2$ form a \textit{strongly unextendible system}, meaning that no unit vector exists unbiased to all of them. Our result implies that adding even just one of $\mathcal{B}_1, \mathcal{B}_2$ to the collection $\mathcal{E}_1, \dots, \mathcal{E}_{p^2-p}$ gives a system admitting no completion.

\section{A common framework for k-nets and MUBs}\label{sec:framework}

We now extend the notion of classical $k$-nets and MUBs to a more general setting, the basis for which is a finite-dimensional $C^*$-algebra.

\subsection{Finite-dimensional $C^*$-algebras}

Let $\mathfrak{A}$ be a complex vector space equipped with a bilinear, associative product (written as multiplication) and a conjugate-linear, anti-automorphic involution $A \mapsto A^*$, i.e.\ let $\mathfrak{A}$ be a complex \textit{*-algebra}. Recall that $\mathfrak{A}$ is a \textit{$C^*$-algebra} if it admits a norm $\norm{\cdot}$ such that  $\norm{AB} \leq \norm{A}\norm{B}$ and
$\norm{A^*A} = \norm{A}^2$
for all $A,B \in \mathfrak{A}$, and $\mathfrak{A}$ is complete with respect to the induced metric $d(A,B) = \norm{A-B}$. (This latter property is of no concern to us since we shall remain in finite dimensions.) $\mathfrak{A}$ is said to be \textit{unital} if it has an element $I \in \mathfrak{A}$ satisfying $IA=AI=A$ for 
all $A \in \mathfrak{A}$; note that, if it exists, $I$ is uniquely determined by this property.

Elements $A$ and $U$ of a $C^*$-algebra $\mathfrak{A}$ satisfying $A = A^*$ and $U^*U = UU^* = I$ are called \textit{self-adjoint} and \textit{unitary}, respectively, while elements $P \in \mathfrak{A}$ with $P = P^2 = P^*$ are called \textit{orthogonal projections}. We say that an orthogonal projection is \textit{minimal} if it cannot be written as the sum of two non-zero orthogonal projections.

Every subalgebra of the set of complex $n \times n$ matrices that is closed under taking adjoints is naturally a $C^*$-algebra (where the notion of an orthogonal projection coincides with the usual geometric one). In fact, every finite-dimensional $C^*$-algebra can be realized this way, as the following characterization confirms (see e.g.\ \cite{davidson96} for a proof).

\begin{proposition}\label{finiteCstar}
Every $d$-dimensional $C^*$-algebra $\mathfrak{A}$ is *-isomorphic to the direct sum of full matrix algebras
\begin{equation*}
    \mathfrak{A} \cong \bigoplus_{j=1}^k M_{n_j}(\C),
\end{equation*}
for some $n_1, \dots, n_k$ satisfying $\sum_{j=1}^k n_j^2 = d$, and all such *-isomorphisms are unitarily equivalent. 
\end{proposition}

In particular, each finite-dimensional $C^*$-algebra $\mathfrak{A}$ is unital and possesses a \textit{canonical trace} $\Tr$; that is, a linear functional $\Tr: \mathfrak{A} \to \C$ satisfying $\Tr(AB) = \Tr(BA)$ for all $A,B \in \mathfrak{A}$ and assigning $1$ to all minimal projections of $\mathfrak{A}$. The \textit{canonical normalized trace} $\tau:\mathfrak{A} \to \C$ is the standard trace scaled such that $\tau(I) = 1$; i.e.\ $\tau = \frac{1}{\Tr(I)} \Tr$.
\vspace{-1mm}

The set $\C^X$ of complex functions on a finite set $X$ also forms a finite-dimensional $C^*$-algebra under pointwise multiplication and conjugation. In this sense, a function $f \in \C^X$ is an orthogonal projection if and only if it is the \textit{indicator function} of a subset; i.e.\ iff $f = \chi_H$ for some $H \subset X$, where $\chi_H$ is the function taking $1$ on elements of $H$ and zero on $X \setminus H$. Note also that the canonical trace of $\C^X$ is simply the summation over all points: $\Tr(f) = \sum_{x \in X} f(x)$.

Every finite-dimensional $C^*$-algebra $\mathfrak{A}$ is naturally an \textit{inner product space} with inner product (sometimes called the \textit{Hilbert-Schmidt inner product}) given by the formula
\begin{equation*}
    \inner{A,B} = \tau(A^*B),
\end{equation*}
where $\tau$ is its canonical normalized trace. This is a consequence of the fact that the canonical trace is positive and faithful, i.e.\ that $\Tr(A^*A) > 0$ for $A \neq 0$. The corresponding norm is
\begin{equation*}
    \norm{A}_2 = \sqrt{\inner{A,A}} = \sqrt{\tau(A^*A)}.
\end{equation*}
In what follows, we shall exploit that this inner product must satisfy the \textit{Cauchy-Schwarz inequality}; that is, for every $X,Y \in \mathfrak{A}$,
\begin{equation*}
    \abs{\inner{X,Y}} \leq \norm{X}_2\norm{Y}_2,
\end{equation*}
with equality holding if and only if $X$ and $Y$ are linearly dependent.

We are frequently concerned with subalgebras of $\mathfrak{A}$ that are closed under the taking adjoints and contain $I \in \mathfrak{A}$, i.e.\! \textit{unital *-subalgebras} of $\mathfrak{A}$. Two such subalgebras $\mathcal{A}, \mathcal{B} \subset \mathfrak{A}$, as linear subspaces, cannot be orthogonal with respect to the Hilbert-Schmidt inner product since $0 \neq I \in \mathcal{A} \cap \mathcal{B}$. We will say, however, that they are \textit{quasi-orthogonal} if their traceless parts $\mathcal{A} \ominus \C I$ and $\mathcal{B} \ominus \C I$ are orthogonal, where
\begin{equation*}
    \mathcal{X} \ominus \C I = \mathcal{X} \cap (\C I)^\perp = \{ X \in \mathcal{X} \mid \tau(X) = 0 \}.
\end{equation*}
It is easy to see that $\mathcal{A}$ and $\mathcal{B}$ are quasi-orthogonal 
if and only if
\begin{equation*}
    \tau(AB) = \tau(A)\tau(B)
\end{equation*}
for all $A \in \mathcal{A}, \, B \in \mathcal{B}$.
For more on quasi-orthogonality, see \cite{petz,weiner10}.

\subsection{MUBs and quasi-orthogonal MASAs}

A maximal set of commuting operators $\mathcal{A} \subset M_n(\C)$ automatically forms a subalgebra 
containing the identity. When $\mathcal{A}$ is also closed under taking adjoints, it is called a 
\textit{maximal abelian *-subalgebra} (MASA). Since $X + X^*$ and $i(X - X^*)$ are always self-adjoint and $X = \frac{1}{2}(X + X^*) - \frac{i}{2} i(X - X^*)$, every *-subalgebra $\mathcal{A} \subset M_n(\C)$ is linearly spanned by its self-adjoints. However, commuting self-adjoint matrices can always be simultaneously diagonalized in some orthonormal basis. Thus, a MASA of $M_n(\C)$ is nothing but the set of all diagonal matrices $\mathcal{A}_\mathcal{E}$ in an orthonormal basis $\mathcal{E}$. Moreover, it is also clear that $\mathcal{A}_\mathcal{E} = \mathcal{A}_\mathcal{F}$ if and only if $\mathcal{E}$ and $\mathcal{F}$ have the same set of coordinate axes.

If $\mathbf{e}$ and $\mathbf{f}$ are unit vectors in $\C^n$ with orthogonal projections onto their one-dimensional subspaces given by $P$ and $Q$, respectively, then $\Tr(PQ) = \abs{\inner{\mathbf{e}, \mathbf{f}}}^2$. Thus, $\mathcal{E} = \{ \mathbf{e}_1, \dots, \mathbf{e}_d \}$ and $\mathcal{F} = \{ \mathbf{f}_1, \dots, \mathbf{f}_d \}$ are mutually unbiased if and only if the respective orthogonal projections onto their axes, $P_1, \dots, P_d$ and $Q_1, \dots, Q_d$, satisfy the relation
\begin{equation*}
    \Tr(P_j Q_l) = 1/d = \tfrac{1}{d} \Tr(P_j) \Tr(Q_\ell)
\end{equation*}
for all $j,\ell$. Since $\mathcal{A}_\mathcal{E}$ and $\mathcal{A}_\mathcal{F}$ are spanned by these projections, it follows that $\mathcal{E}$ and $\mathcal{F}$ are mutually unbiased if and only if $\mathcal{A}_\mathcal{E}$ and $\mathcal{A}_\mathcal{F}$ are quasi-orthogonal. Thus, one can consider collections of (pairwise) quasi-orthogonal MASAs instead of collections of MUBs, as is often done in the literature  
\cite{petz, weiner13,krusekamp14, szanto16}.

\subsection{Generalized $k$-nets}

Having established the necessary background, we can finally give a formal definition to our generalization.

\begin{definition}
Let $\mathfrak{A}$ be a finite-dimensional $C^*$-algebra with canonical normalized trace $\tau$. We shall say that a collection of orthogonal projections $\mathcal{N} \subset \{P \in \mathfrak{A} \mid P^2 = P^* = P\}$ is a \textit{$k$-net over $\mathfrak{A}$}, if it satisfies the following properties:
\begin{enumerate}[(i)]
    \item the relation ``$P = Q$ or $PQ = 0$'' is an equivalence relation on $\mathcal{N}$ dividing
    $\mathcal{N}$ into $k$ equivalence classes;
    \item if $P,Q \in \mathcal{N}$ fall into different equivalence classes, then \\ $\tau(PQ) = \frac{1}{\dim(\mathfrak{A})}$;  
    \item the elements in each equivalence class sum to the identity $I$.
\end{enumerate}
\end{definition}

In the context of $k$-nets, we shall refer to elements of $\mathcal{N}$ as \textit{lines} and to the introduced equivalence classes as \textit{parallel classes}. If $k$ is at least two and $P_1, \dots, P_q$ are the lines of
one parallel class while $Q_1, \dots, Q_r$ are the lines of another, then
\begin{equation*}
    \tau(P_j) = \tau(P_j I) = \tau\Big(P_j \sum_{\ell=1}^r Q_\ell\Big) =  \sum_{\ell=1}^r \tau(P_j Q_l) = \frac{r}{\dim(\mathfrak{A})},
\end{equation*}
showing that the trace of the lines within a class is constant. Since this trace value also gives the number of elements in any other parallel class, we have that, for $k \geq 3$, each parallel class must have the same number $d$ of lines and
\begin{equation*}
    1 = \tau(I) = \sum_{j,\ell=1}^d \tau(P_j Q_\ell) = \frac{d^2}{\dim(\mathfrak{A})},
\end{equation*}
showing that $\dim(\mathfrak{A}) = d^2$. When each parallel class has the same number $d$ of lines (which we just noted was automatic for $k\geq 3$), we shall say that $\mathcal{N}$ is a \textit{$k$-net of order $d$}. By now, it is trivial to observe that this definition of $k$-nets over finite dimensional $C^*$-algebras generalizes the notions of both classical $k$-nets and MUBs. However, we feel that this is worth stating more formally (although we omit a proof).

\begin{proposition}
Let $X$ be a finite set of points and $n$ a natural number.
\begin{itemize}
    \item $\mathcal{N}$ is a $k$-net of order $d$ over $\C^X$ if and only if it is the set of indicator functions of the lines of a (classical) $k$-net of order $d$ on X (and in this case $|X|=d^2$).
    \item $\mathcal{N}$ is a $k$-net of order $d$ over $M_n(\C)$ if and only if it is the set of orthogonal projections onto the axes of the bases of a collection of MUBs (and in this case $n=d$).
\end{itemize}
\end{proposition}

As is well-known in the MUB case, it also holds in this general setting that a $k$-net of order $d$ over an algebra $\mathfrak{A}$ determines a collection of $k$ quasi-orthogonal, $d$-dimensional, abelian, and unital *-subalgebras of $\mathfrak{A}$. Indeed, let $\mathcal{A}_j$ be the linear subspace spanned by the lines of the $j$th parallel class. By the third defining property of $k$-nets, $I \in \mathcal{A}_j$. As it is spanned by self-adjoint elements, $\mathcal{A}_j$ is also closed under taking adjoints. Since, within a parallel class, the product of two distinct lines is zero and the square of any line is itself, $\mathcal{A}_j$ is also closed under multiplication. But for projections $P$ and $Q$, the relation $PQ = 0$ implies $QP = 0$, so $\mathcal{A}_j$ is actually a commutative algebra. Moreover, since its spanning projections are mutually orthogonal (and nonzero), the dimension of this algebra coincides with the number $d$ of lines in the $j$th parallel class. Finally, if $A \in \mathcal{A}_\ell$ and $B \in \mathcal{A}_{\ell'}$ for $\ell \neq \ell'$, then we have scalars $\lambda_1, \dots, \lambda_d$ and $\mu_1, \dots, \mu_d$ such that
\begin{equation*}
    A = \lambda_1 P_1 + \dots + \lambda_d P_d, \quad B = \mu_1 Q_1 + \dots + \mu_d Q_d,
\end{equation*}
where the $P$ and $Q$ operators are the lines of the two parallel classes. Hence,
\begin{align*}
    \tau(AB) &= \tau\big(\big(\lambda_1 P_1 + \dots + \lambda_d P_d\big)\big(\mu_1 Q_1 + \dots + \mu_d Q_d\big)\big)\\
    &= \tfrac{1}{d^2}(\lambda_1 + \dots + \lambda_d)(\mu_1 + \dots + \mu_d) = \tau(A)\tau(B),
\end{align*}
so the subalgebras are indeed pairwise quasi-orthogonal. Let us see now what we can say about the converse. 

A *-subalgebra $\mathcal{A}$ of a finite-dimensional  $C^*$-algebra is linearly spanned by its minimal projections. In the commutative case, the product of two orthogonal projections $P,Q \in \mathcal{A}$ is automatically an orthogonal projection as well; thus, if $P \neq Q$ are minimal, then $PQ = 0$. If $\mathcal{A}$ is also unital, its minimal projections must sum to the identity and $\dim(\mathcal{A})$ is simply their count. Suppose now that $\mathcal{A}$ and $\mathcal{B}$ are two such subalgebras with minimal projections $P_1, \dots, P_d$ and $Q_1, \dots, Q_d$, respectively. If $\mathcal{A}$ and $\mathcal{B}$ are quasi-orthogonal, then $\tau(P_j Q_\ell)=\tau(P_j)\tau(Q_\ell)$; however, in general it does not necessarily follow that $\tau(P_j Q_\ell) = 1/d^2$. Nevertheless, it is not difficult to see that it does so in our two cases of interest: when $\mathfrak{A}$ is isomorphic to $M_d(\C)$ or $\C^X$. Regardless, the fact that $k$-nets of order $d$ give rise to $k$ quasi-orthogonal $d$-dimensional subalgebras implies that the usual argument regarding the maximum possible $k$ value can be repeated in general.

\begin{proposition}\label{prop:kNetBound}
Let  $\mathcal{N}$ be a $k$-net of order $d$ over $\mathfrak{A}$. Then $k \leq d + 1$, and equality holds if and only if $\mathcal{N}$ linearly spans $\mathfrak{A}$, in which case we shall say that that $\mathcal{N}$ is \textit{complete}.
\end{proposition}

\begin{proof}
Consider the $k$ quasi-orthogonal, unital, $d$-dimensional *-sub\-al\-geb\-ras associated with the parallel classes of $\mathcal{N}$. Their traceless parts are orthogonal ($d-1$)-dimensional subspaces, so $\mathcal{N}$ spans a   $k(d-1) + 1$ dimensional space. The claim then follows because $\dim(\mathfrak{A})=d^2$.
\end{proof}

\section{A general rigidity theorem}\label{sec:rigidity}

We now have the framework and background in place to establish our primary result and its important corollary. Suppose $\mathfrak A$ is a finite dimensional $C^*$-algebra with normalized
canonical trace $\tau$.

\begin{theorem}\label{thm:rigidity} 
Let $\mathcal{N}$ be a $k$-net of order $d$ over $\mathfrak{A}$, and suppose that $k \leq \sqrt{d}$.
If $P\!=\!P^2\!=\!P^*\! \in {\rm Span}(\mathcal{N})$ with $\tau(P)=\frac{1}{d}$, then $P \in \mathcal{N}$.
\end{theorem}

\begin{corollary}\label{cor:completionUniqueness}
Let $\mathcal{N}$ be a $(d\!+\!1\!-\!k)$-net of order $d$ over $\mathfrak{A}$
with $k \leq \sqrt{d}$. If $\mathcal{N}$ can be completed to a full $(d\!+\!1)$-net $\widetilde{\mathcal{N}}$, then this completion is unique, and its extra lines are precisely the ``transversals'' of $\mathcal{N}$; i.e. $Q\in \widetilde{\mathcal{N}} \setminus \mathcal{N}$ if and only if 
$Q\!=\!Q^2\!=\!Q^*$ with
$\tau(QP) = \frac{1}{d^2}$ for all $P \in \mathcal{N}$.
\end{corollary}

Concretely, this means that sufficiently large combinatorial $k$-nets and sets of MUBs can be completed in at most one way.

\begin{proof}[Proof of Theorem \ref{thm:rigidity}]

First, note that the case of $k=1$ is trivial, so we may assume that $k \geq 2$ and $d \geq 4$. Now, label the subalgebras of $\mathfrak{A}$ corresponding to the parallel classes of $\mathcal{N}$ by $\mathcal{A}_1, \dots, \mathcal{A}_k$. We can restate the fact that $P$ lies in the span of the lines of $\mathcal{N}$ as
\begin{equation}
    P \in \sum_{j=1}^k \mathcal{A}_j.
\end{equation}
The quasi-orthogonality requirement means that the traceless parts of these subalgebras are pairwise orthogonal, so we can uniquely express the traceless part of $P = P^*$ as
\begin{equation}\label{eq:decomp}
    P - \frac{1}{d}I = \sum_{j=1}^k A_j,
\end{equation}
where each $A_j \in \mathcal{A}_j \ominus \C I$ with $A_j = A_j^*$.  Observing that the $A_j$ operators are pairwise orthogonal 
(with respect to the Hilbert-Schmidt inner product)
and that $P^2 = P$, we have
\begin{align}
    \sum_{j=1}^k \norm{A_j}_2^2 &= \norm*{P -\tfrac{1}{d}I}_2^2 = \tau\left(\left( P -\tfrac{1}{d}I\right)^2\right)\\ &= \tau\left(\frac{d-2}{d}P + \frac{1}{d^2}I\right) = \frac{d-1}{d^2}.
\end{align}
With this in mind, we introduce the length ratios
\begin{equation}
    t_j = \frac{\norm{A_j}_2}{\norm*{P - \frac{1}{d}I}_2} = \frac{d}{\sqrt{d-1}}\norm{A_j}_2
\end{equation}
to describe how $P - \frac{1}{d}I$ is distributed among the $\mathcal{A}_j \ominus \C I$ subspaces, noting that $\sum_{j=1}^k t_j^2 = 1$. If any $t_j = 1$, then $P \in \mathcal{A}_j$, and if any $t_j = 0$, then the subalgebra $\mathcal{A}_j$ is unnecessary and can be omitted. From now on, we therefore assume that each $t_j \in (0, 1)$
%%while $k\leq \sqrt{d}$
and aim to reach a contradiction by examining the spectrum of $A_r$, where $r$ is an index such that $t^2_r$ is at least the average value of $\frac{1}{k}$. Returning to \eqref{eq:decomp} and squaring both sides, we find
\begin{equation}
    \sum_{j,\ell} A_j A_\ell = \frac{d-2}{d} \sum_j A_j - \frac{d-1}{d^2}I.
\end{equation}
Next, we take the inner product of this equation with $X^*$ for some traceless ${X \in \mathcal{A}_r \ominus \C I}$. Since the subalgebras are quasi-orthogonal and closed under multiplication, $\inner{X^*,A_j A_\ell} = \tau(X A_j A_\ell)$ vanishes when precisely two of $r$, $j$, and $\ell$ coincide. After rearrangement, we obtain
\begin{equation}\label{eq:X}
    \sum_{j \neq r} \sum_{\ell \neq r,j} \tau(X A_j A_l) = \frac{d-2}{d} \tau(X A_r) - \tau(X A_r^2).
\end{equation}
We now use this equality to uncover a gap in the spectrum of $A_r$.

\begin{lemma}
Each $\lambda \in \Sp(A_r)$ satisfies $\lambda \leq \lambda_-$ or $\lambda \geq \lambda_+$, where
\begin{equation*}
\lambda_\pm = \frac{1}{2d}\Bigg(d - 2 \pm \sqrt{\big(d-2\big)^2 - 4\big(d-1\big)\big(k - 3 + \tfrac{1}{k}\big)}\Bigg).
\end{equation*}
\end{lemma}

\begin{proof}
Because $\mathcal{A}_r$ is the span of $d$ lines in a parallel class of $\mathcal{N}$, we have that
\begin{equation*}
    A_r = \lambda_1 Q_1 + \dots + \lambda_d Q_d,
\end{equation*}
where each $Q_j = Q_j^* = Q_j^2 \in \mathcal{A}_r$ with $\tau(Q_j) = \frac{1}{d}$ and $Q_j Q_\ell = 0$ for $j \neq \ell$.
This is a spectral decomposition of a self-adjoint operator, so $\Sp(A_r) = \{ \lambda_1, \dots, \lambda_d \} \subset \R$.

Now, fix some $\lambda=\lambda_s$ with corresponding projection $Q=Q_s$. Substituting $Q - \tfrac{1}{d}I$, the traceless part of $Q$, for $X$ in \eqref{eq:X} and noting the quasi-orthogonality of the $\mathcal{A}_j$ subalgebras, we arrive at
\begin{equation}\label{eq:Q}
    \sum_{j \neq r}\sum_{\ell \neq r,j} \tau(Q A_j A_\ell) = \frac{d-2}{d^2} \lambda - \frac{1}{d} \lambda^2 + \frac{d-1}{d^3} t_r^2.
\end{equation}
Applying the triangle inequality gives upper bound
\begin{equation*}
    \sum_{j \neq r}\sum_{\ell \neq r,j} \tau(Q A_j A_\ell) \leq \sum_{j \neq r}\sum_{\ell \neq r,j} \abs{\tau(Q A_j A_\ell)},
\end{equation*}
where
\begin{align*}
    \abs*{\tau(Q A_j A_\ell)} = \abs*{\tau(Q A_j A_\ell Q)} &\leq \tau(Q A_j^2) \tau(Q A_\ell^2)\\
    &= \frac{1}{d^2} \tau(A_j^2) \tau(A_\ell^2) = \frac{(d-1)^2}{d^6} t_j^2 t_\ell^2,
\end{align*}
using $Q^2 = Q$, Cauchy-Schwartz, and quasi-orthogonality. With this inequality, we bound the LHS of \eqref{eq:Q} from above by
\begin{align*}
     \sum_{j \neq r}\sum_{\ell \neq r,j} \tau(Q A_j A_\ell) &\leq \frac{d-1}{d^3} \sum_{j \neq r}\sum_{\ell \neq r,j} t_j t_\ell = \frac{d-1}{d^3}\left(\bigg(\sum_{j \neq r} t_j\bigg)^2 - \sum_{j \neq r} t_j^2\right)\\
     &\leq \frac{d-1}{d^3}\Bigg(\left((k-1)\sqrt{\frac{1-t_r^2}{k-1}}\,\right)^2 - \big(1 - t_r^2\big)\Bigg)\\
     &= \frac{d-1}{d^3}\big(k-2\big)\big(1 - t_r^2\big).
\end{align*}
Applying this bound to the RHS of \eqref{eq:Q} and rearranging gives
\begin{equation*}
    \lambda^2 - \frac{d-2}{d} \lambda + \frac{d-1}{d^2}\big(\big(k-2\big)-\big(k-1\big) t_r^2\big) \geq 0.
\end{equation*}
Finally, we utilize $t_r^2 \geq \tfrac{1}{k}$ and $k \geq 1$ to obtain
\begin{equation}
    \lambda^2 - \frac{d-2}{d} \lambda + \frac{d-1}{d^2}\left(k - 3 +\frac{1}{k}\right) \geq 0.
\end{equation}
The discriminant of this quadratic is non-negative for $k \leq \sqrt{d}$, so we have our desired result.
\end{proof}

Let us first consider the case where there exists some large $\lambda \in \Sp(A_r)$ with $\lambda \geq \lambda_+$. We will reach a contradiction by returning to \eqref{eq:X} and substituting $A_r$ for $X$, which gives
\begin{equation}\label{eq:Ar}
    \sum_{j \neq r} \sum_{\ell \neq r,j} \tau(A_r A_j A_\ell) = \frac{d-2}{d} \tau(A_r^2) - \tau(A_r^3).
\end{equation}
To bound \eqref{eq:Ar} from below, we find a theoretical maximum for $\tau(A_r^3)$ subject to $A_r = A_r^*$, $\tau(A_r) = 0$, and $\tau(A_r^2) = \frac{d-1}{d^2} t_r^2$ with fixed $t_r$. The method of Lagrange multipliers reveals that this maximum is achieved by an extreme
Click to hide the PDF
 spectrum for $A_r$ with one large, positive eigenvalue of $\tfrac{d-1}{d} t_r$ and many small, negative eigenvalues equal to $-\tfrac{1}{d} t_r$. Straightforward computation then provides the bound 
\begin{equation}\label{eq:ArLowerBound}
    \sum_{j \neq r} \sum_{\ell \neq r,j} \tau(A_r A_j A_\ell) \geq \frac{(d-2)(d-1)}{d^3}(t_r^2 - t_r^3).
\end{equation}
We note that this estimate is suitable to our current situation because the large eigenvalue of $A_r$ forces a similarly extreme spectrum. Returning to the LHS of \eqref{eq:Ar}, we note that
\begin{equation*}
    \abs[\Big]{\sum_{\ell \neq r,j} \tau\left(A_r A_j A_\ell\right)}^2
    \leq \norm{A_j A_r}_2^2 \cdot \norm[\Big]{\sum_{\ell \neq r,j} A_\ell}_2^2,
\end{equation*}
where
\begin{equation*}
    \norm{A_j A_r}_2^2 = \tau(A_j^2A_r^2) = \tau(A_j^2)\tau(A_r^2) = \frac{(d-1)^2}{d^4}t_j^2 t_r^2
\end{equation*}
and
\begin{equation*}
    \norm[\Big]{\sum_{\ell \neq r,j} A_\ell}_2^2 = \sum_{\ell \neq r,j} \norm{A_\ell}_2^2 = \frac{d-1}{d^2}(1 - t_j^2 - t_r^2).
\end{equation*}
Putting this all together, we bound \eqref{eq:Ar} from above by
\begin{align}\label{eq:ArUpperBound}
\begin{split}
    \sum_{j \neq r} \sum_{\ell \neq r,j} \tau(A_r A_j A_\ell) &\leq \sum_{j \neq r} \abs[\Big]{\sum_{\ell \neq r,j} \tau(A_r A_j A_\ell)}\\
    &\leq \frac{(d-1)^{3/2}}{d^3}t_r \sum_{j \neq r} t_j \sqrt{1 - t_r^2 - t_j^2}\\
    &\leq \frac{\sqrt{(k-2)(d-1)^3}}{d^3} t_r (1 - t_r^2).
\end{split}
\end{align}
The final inequality is obtained by fixing $t_r$ and using Lagrange multipliers to see that the sum is maximized when all summands are equal, with $t_j^2 = \frac{1}{k-1}(1-t_r^2)$ for $j \neq r$. Combining this with \eqref{eq:ArLowerBound}, we find
\begin{equation*}
    \frac{\sqrt{(k-2)(d-1)^3}}{d^3} t_r (1 - t_r^2) \geq \frac{(d-2)(d-1)}{d^3}(t_r^2 - t_r^3).
\end{equation*}
Since $t_r \in (0, 1)$, we can divide by $t_r,1-t_r$ and $1+t_r$ to obtain
\begin{equation}\label{eq:tr}
    \frac{\sqrt{(d-1)(k-2)}}{d-2}  \geq \frac{t_r}{1+t_r}.
\end{equation}
Next, we minimize $t_r$ subject to the existence of our single large eigenvalue. As before, we find the minimum to be achieved by an extreme spectrum for $A_r$, with the other eigenvalues small and negative. This gives
\begin{align*}
t_r &= \sqrt{\frac{d}{d-1}\sum_{j=1}^d\!\lambda_j^2} \geq \sqrt{\frac{d}{d-1}\left(\lambda_+^2 + (d-1)\left(\frac{-\lambda_+}{d-1}\right)^2\right)} = \frac{d\lambda_+}{d-1}.
\end{align*}
Lastly, noting that the function $t \mapsto t/(1+t)$ is monotonically increasing on $(0,1)$, our last two inequalities imply
\begin{equation*}
    \frac{\sqrt{(d-1)(k-2)}}{d-2}
    \geq \frac{\gamma}{1+\gamma},
\;\;\;\;
\mathrm{where}\;\;
\gamma=\frac{d \lambda_+}{d-1}.
\end{equation*}
However, as we prove explicitly in Lemma \ref{lem:inequality1} (see Appendix), this inequality cannot hold unless $k > \sqrt{d}$. Indeed, this is clear in the limit; for the extreme case of $k = \sqrt{d}$, $\lambda_+$ and $t_r$ approach 1 as $d \to \infty$, and the inequality nears $\frac{1}{2} \leq d^{-1/4}$. Hence, we have reached a contradiction.

This leaves us with the final possibility that each $\lambda \in \Sp(A_r)$ satisfies $\lambda \leq \lambda_-$. We return to \eqref{eq:Ar}, but this time, restricted to only small eigenvalues, we observe that
\begin{equation}
\label{eq:nonopt}
    \tau(A_r^3) \leq \lambda_{\max} \tau(A_r^2) \leq \lambda_{-} \tau(A_r^2)
\end{equation}
and bound the RHS of \eqref{eq:Ar} from below by
\begin{equation*}
\frac{d-2}{d} \tau(A_r^2) - \tau(A_r^3) \geq \frac{d-1}{d^2}t_r^2\left(\frac{d-2}{d} - \lambda_-\right).
\end{equation*}
Now, we combine this inequality with our previous upper bound given by \eqref{eq:ArUpperBound} to find
\begin{equation*}
     \frac{\sqrt{(k-2)(d-1)^3}}{d^3} t_r (1 - t_r^2) \geq 
     \frac{d-1}{d^2}t_r^2\left(\frac{d-2}{d} - \lambda_-\right),
\end{equation*}
which simplifies to
\begin{equation*}
\frac{\sqrt{(d-1)(k-2)}}{d-2-d\lambda_-} \geq \frac{t_r}{1-t_r^2} .
\end{equation*}
Finally, since the function $t \mapsto t/(1-t^2)$ is monotonically increasing on $(0,1)$, and $t_r$ was chosen to be at least $\tfrac{1}{\sqrt{k}}$, we can deduce with some rearrangement that
\begin{equation}
    \frac{\sqrt{k-2}(k-1)}{\sqrt{k}} \geq \frac{d-2-d\lambda_-}{\sqrt{d-1}}.
\end{equation}
Once again, this inequality can only be satisfied if $k > \sqrt{d}$, as we compute explicitly in Lemma \ref{lem:inequality2} (see Appendix). Thus, we have reached a final contradiction, and the theorem is complete.
\end{proof}
From here, the corollary is straightforward.
\begin{proof}[Proof of Corollary \ref{cor:completionUniqueness}]
Let $\mathcal{N}$ be a $(d+1-k)$-net of order $d$ over $\mathfrak{A}$ with $k \leq \sqrt{d}$, and denote the subalgebras spanned by its parallel classes by $\mathcal{B}_1, \dots, \mathcal{B}_{d+1-k}$. If $P$ is an orthogonal projection satisfying the trace requirement appearing in the claim, then  $P-\frac{1}{d}I$ is orthogonal to each subalgebra $\mathcal{B}_j$. Equivalently, $P$ must be an orthogonal projection of ``size'' $\tau(P)=\frac{1}{d}$ lying in the subspace sum
\begin{equation*}
    \mathcal{X} = \C I + \Big(\sum_j\mathcal{B}_j\Big)^\perp,
\end{equation*}
which has dimension $d^2 - (d-1)(d+1-k) = k(d-1) + 1$. If $\mathcal{N}$ can be completed to a full $(d+1)$-net $\widetilde{\mathcal{N}}$, then elements of the new parallel classes $\mathcal{M} = \widetilde{\mathcal{N}} \setminus \mathcal{N}$ form a $k$-net of order $d$ over $\mathfrak{A}$. Labeling the corresponding subalgebras of $\mathfrak{A}$ by $\mathcal{A}_1, \dots, \mathcal{A}_k$, the subspace sum $\sum_j \mathcal{A}_j$ (equivalently, the span of the elements of $\mathcal{M}$) must be contained in $\mathcal{X}$. However, for Proposition \ref{prop:kNetBound}, we showed that this subspace has dimension $k(d-1) + 1$, so it must equal $\mathcal{X}$. Therefore, any orthogonal projection $P \in \mathfrak{A}$ satisfying the trace condition of the claim  must lie in the span of the lines of $\mathcal{M}$, so Theorem \ref{thm:rigidity} gives that $P \in \mathcal{M}$.
\end{proof}

\begin{remark}
\label{betterboundremark}
The bound given on $k$ in our main
theorem could be slightly improved.
For example, \eqref{eq:nonopt} was nice and simple to use, but it is clearly suboptimal; since $\tau(A_r) = 0$, it must have some negative eigenvalues, so $\tau(A_r)^3$ must be strictly smaller than $\lambda_{\max}\tau(A_r^2)$. 
Also, when $k=2$, the LHS of \eqref{lem:inequality2} is precisely zero, immediately showing that for $k=2,d\neq 2$ the projection $P$ must be an element of $\mathcal{N}$.
Taking account of these and the integrality of $k$, we actually numerically justified with a computer that at least up to $d=1000$, the
condition $k \leq \sqrt{d}$  could be replaced by $k< \sqrt{d}+1$ with the exception of six orders (all between $10$ and $19$) needing
a further case-by-case study, something 
clearly unproportional to the
possible gain.
\end{remark}

\section{Nice mutually unbiased bases and uncompletability}\label{sec:niceMUBs}

We have just demonstrated that sufficiently large sets of MUBs have at most one complete extension. Now, it is natural to examine the structure of this unique completion if we place certain conditions on the initial set. Hence, we next restrict ourselves to a subset of so-called nice MUBs, which have a convenient algebraic description. With these objects, we will prove a stricter rigidity theorem that implies certain sets of MUBs cannot be completed.

\subsection{Nice MUBs}
In quantum information theory, orthonormal bases of unitary matrices, called \textit{unitary error bases}, are fundamental to error correction and super-dense coding \cite{klappenecker03}. These unitaries are often constructed algebraically, motivating the following definition \cite{klappenecker05}.

\begin{definition}
Let $G$ a group of order $d^2$ with identity $e$. A \textit{nice error basis}, also called a unitary operator basis of group type, is a set $\mathcal{E} = \{ U(g) \mid g \in G \}$ of unitary operators in $M_d(\C)$ such that
\begin{enumerate}[(i)]
    \item $U(e) = I$,
    \item $\Tr(U(g)) = 0$ for $e \neq g \in G$,
    \item $U(g)U(h) = \lambda(g,h)U(gh)$ for all $g,h \in G$,
\end{enumerate}
where $\lambda(g,h)$ is a complex phase factor.
\end{definition}

By (i) and (iii), $U$ determines a projective representation of $G$. Evaluating (iii) at $g$ and $g^{-1}$ gives that $U(g)^* = \lambda(g^{-1},g)^{-1}U(g^{-1})$. From this and (ii), we find that $\Tr(U(g)^*U(h)) = 0$ for all $g \neq h$. Thus, as the name suggests, nice error bases are a special class of unitary error bases. In practice, many nice error bases are some variant of the following, described in \cite{bandyopadhyay02}.

\begin{example}\label{ex:discreteWeyl}
Fix $d \geq 2$, and let $\omega = e^{2\pi i / d}$. We define $X_d$ to be the cyclic shift matrix $X_d \, \mathbf{e}_j = \mathbf{e}_{j+1} \pmod{d}$ and $Z_d$ to be the diagonal matrix $Z_d \, \mathbf{e}_j = \omega^{j-1} \mathbf{e}_j$, where $\{ \mathbf{e}_j \}_j$ is the standard basis. Then, the \textit{discrete Weyl operators} $\{ X_d^jZ_d^\ell \mid (j,\ell) \in \Z_d^2 \}$ form a nice error basis with index group $\Z_d^2$.
\end{example}

The group $G$ is called the \textit{index group} of $\mathcal{E}$. For each subgroup $H$ of the index group, define
\begin{equation*}
\mathcal{A}_H := \vspan\{ U(h) \mid h \in H \}
\end{equation*}
to be the subspace of $M_d(\C)$ spanned by the unitaries corresponding to $H$. By (iii), $\mathcal{A}_H$ is actually a *-subalgebra, closed under operator multiplication and taking adjoints. If $H$ is of order $d$ and its associated unitaries are pairwise commuting, then $\mathcal{A}_H$ is a MASA of $M_d(\C)$, corresponding to an orthonormal basis of $\C^d$.

If two subgroups $H,H' \leq G$ have trivial intersection $H \cap H' = \{e\}$, then the operator orthogonality of the unitaries in $\mathcal{E}$ implies that $\mathcal{A}_H$ and $\mathcal{A}_{H'}$ are quasi-orthogonal subalgebras. Putting this all together, we have the following well-known result \cite{aschbacher07}.

\begin{proposition}\label{prop:niceMUBs}
Let $\mathcal{E}$ be a nice error basis for $M_d(\C)$ with index group $G$, and take $H_1, \dots, H_m$ to be subgroups of $G$ of order $d$ with pairwise trivial intersections. If, for each $H_j$, the associated unitaries of $\mathcal{E}$ are pairwise commuting, then $\mathcal{A}_{H_1}, \dots, \mathcal{A}_{H_m}$ are quasi-orthogonal MASAs of $M_d(\C)$, corresponding to a set of MUBs.
\end{proposition}

We call bases constructed from a nice error basis $\mathcal{E}$ in this way $\mathcal{E}$-\textit{nice mutually unbiased bases}. Nice MUBs are well-studied for their convenient algebraic structure, although they are insufficient to resolve the prime power conjecture for MUBs in general \cite{aschbacher07}. Our second theorem concerns nice MUBs with abelian index groups.

If $\mathcal{E} = \{ U(g) \mid g \in G\}$ is a nice error basis with abelian index group $(G,+)$, then
\begin{equation*}
U(g)U(h) = \lambda(g,h)U(g + h) = \lambda(g,h)\lambda(h,g)^{-1}U(h)U(g).
\end{equation*}
Hence, we define the \textit{commutator map} $\sigma:G \times G \to \C$ of $\mathcal{E}$ by
\begin{equation*}
\sigma(g,h) = \lambda(g,h)\lambda(h,g)^{-1},
\end{equation*}
so that $U(g)$ and $U(h)$ commute if and only if $\sigma(g,h) = 1$. For a given subgroup $H \leq G$, the corresponding the subalgebra $\mathcal{A}_H$ is abelian exactly when $\sigma|_{H \times H} \equiv 1$. This allows us to express Proposition \ref{prop:niceMUBs} more concretely when $G$ is abelian.

\begin{proposition}
Let $\mathcal{E}$ be a nice error basis for $M_d(\C)$ with abelian index group $G$ and commutator map $\sigma$. If $H_1, \dots, H_m$ are order $d$ subgroups of $G$ with pairwise trivial intersections such that $\sigma|_{H_j \times H_j} \equiv 1$ for each $H_j$, then $\mathcal{A}_{H_1}, \dots, \mathcal{A}_{H_m}$ are quasi-orthogonal MASAs of $M_d(\C)$, corresponding to a set of MUBs.
\end{proposition}

We can now prove that the bound of $k \leq \sqrt{d}$ in Theorem \ref{thm:rigidity} is tight when $d$ is the square of an odd prime. The following construction was inspired by a similar one appearing in the work of Sz\'{a}nt\'{o} \cite{szanto16}.

\begin{example}\label{ex:discreteWeylPrimeSquare}
Consider the unitaries of the form
\begin{equation*}
    U(j,\ell,r,s) = X_d^jZ_d^\ell \otimes X_d^rZ_d^s,
\end{equation*}
where $X_d$ and $Z_d$ are the discrete Weyl operators introduced in Example \ref{ex:discreteWeyl}. The collection $\mathcal{E} = \{ U(j,\ell,r,s) \mid j,\ell,r,s \in \Z_d \}$ is a nice error basis for $M_d(\C) \otimes M_d(\C) \cong M_{d^2}(\C)$ with abelian index group $\Z_d^4$. Taking $d=p$ to be an odd prime and $D \in \Z_p$ to be a quadratic nonresidue, the subgroups
\begin{align*}
    R_{x,y} &= \inner{(1,x,y,0), (0,-y,-Dx,1)}, \quad x,y \in \Z_p\\
    R_\infty &= \inner{(0,1,0,0),(0,0,1,0)}
\end{align*}
of $\Z_p^4$ generate a complete set of $p^2 + 1$ $\mathcal{E}$-nice MUBs. Furthermore, the MASA of $M_{p^2}(\C)$ corresponding to
\begin{equation*}
    S = \inner{(1,0,0,0),(0,0,1,0)}
\end{equation*}
is contained in the subspace sum of the $p+1$ MASAs $\mathcal{A}_{R_\infty}$ and $\mathcal{A}_{R_{0,y}}$, for $y \in \Z_p$, but is not equal to any of them, providing a tight counterexample\footnote{It is actually not difficult to give a tight counterexample with $p=2$, too -- but not with the above construction, since there is no quadratic nonresidue in $\Z_2$. We shall not go into the details here as we do not use this, but by \cite{pszw} one 
sees that the orthogonal of $\mathcal B + \mathcal B'\ominus \C I$
in $M_4(\mathbb C)\equiv M_2(\mathbb C)\otimes M_2(\mathbb C)$, where
$\mathcal B = M_2(\mathbb C)\otimes I$ and $\mathcal B'=I\otimes M_2(\mathbb C)$, can be actually decomposed into the sum of 3 quasi-orthogonal MASAs  in continuously many different ways.} for Theorem \ref{thm:rigidity} with $k = \sqrt{d} + 1$, $d=p^2$. 
\end{example}

\begin{proof}
This tensor product construction for $\mathcal{E}$ is well-known \cite{bandyopadhyay02}, with the commutator map $\sigma:\Z_p^4 \times \Z_p^4 \to \C$ given by
\begin{equation*}
    \sigma(u,v) = \omega^{u_1v_2 - u_2v_1 + u_3v_4 - u_4v_3},
\end{equation*}
where $\omega = e^{2\pi i / p}$. It is straightforward to check that $\sigma \equiv 1$ when restricted to each subgroup. Clearly, $R_\infty$ has order $p^2$ and intersects the other subgroups trivially. If
\begin{equation*}
    \alpha (1,x_1,y_1,0) + \beta (0, -y_1, -Dx_1, 1) = \tilde{\alpha} (1,x_2,y_2,0) + \tilde{\beta} (0, -y_2, -Dx_2, 1),
\end{equation*}
then $\alpha = \tilde{\alpha}$ and $\beta = \tilde{\beta}$, so each $R_{x,y}$ subgroup also has order $p^2$. Reformulating, we find that
\begin{equation*}
    \begin{pmatrix}
    x_1 - x_2 & y_2 - y_1 \\
    y_1 - y_2 & D(x_2 - x_1)
    \end{pmatrix}
    \begin{pmatrix}
    \alpha\\\beta
    \end{pmatrix} = 0.
\end{equation*}
Since $D$ is not a quadratic residue, this matrix is non-singular unless $x_1 = x_2$ and $y_1 = y_2$, forcing $\alpha = \beta = 0$. Thus, the $R_{x,y}$ subgroups have pairwise trivial intersections. All together, this gives that the $R_{x,y}$ subgroups along with $R_\infty$ indeed generate a set of $p^2 + 1$ $\mathcal{E}$-nice MUBs.

Finally, it is easy to verify that $S$ is a distinct subgroup of order $p^2$, and it is contained in the union of $R_\infty$ and each $R_{0,y}$. Thus, $\mathcal{A}_S$ is a MASA of $M_{p^2}(\C)$ contained in the subspace sum of $\mathcal{A}_{R_\infty}$ and each $\mathcal{A}_{R_{0,y}}$.
\end{proof}

\begin{remark}
We note that the definitions of this subsection can be extended naturally to the commutative case, but we omit a discussion of ``nice $k$-nets'' because our next results are fundamentally tied to the non-commutative nature of $M_d(\C)$.
\end{remark}

\subsection{A nice rigidity theorem}
The main observation from which we develop our second theorem is that an abelian index group of a nice error basis acts via the conjugation action in a well-behaved way on many subalgebras of interest. We formalize this with the following lemma.

\begin{lemma}\label{le1}
Let $\mathcal{E} = \{U(g) \mid g \in G\}$ be a nice error basis for $M_d(\C)$ with abelian index group. Suppose that $\mathcal{C}$ is a collection of at least $d+1-\sqrt{d}$ $\mathcal{E}$-nice MUBs that can be completed. If $H$ is a subgroup of $G$ corresponding to a basis of $\mathcal{C}$ and $\mathcal{B}$ is a subalgebra of $M_d(\C)$ corresponding to a basis of $\mathcal{C}$'s completion, then $H$ acts transitively on the rank one orthoprojections of $\mathcal{B}$ via the conjugation action defined by $Q \xmapsto{a} U(a)QU(a)^*$ for $a \in H$.
\end{lemma}

\begin{proof}
Denote the subgroups of $G$ generating $\mathcal{C}$ by $H_1, \dots, H_{d+1-k}$, where $k \leq \sqrt{d}$. By Theorem \ref{thm:rigidity}, we know that the completion of $\mathcal{C}$ is unique; label the subalgebras of $M_d(\C)$ corresponding to the completion bases by $\mathcal{B}_1, \dots, \mathcal{B}_k$. Fixing arbitrary $H_j$ and $\mathcal{B}_\ell$, we must first show that the conjugation action of $H_j$ on the rank one projections of $\mathcal{B}_\ell$ is well-defined. Clearly, $e \in H_j$ corresponds to the identity map, so we focus on non-identity $a \in H_j$. Note that each $\mathcal{A}_{H_{j'}}$ is an invariant subspace under $a$, i.e.\
\begin{equation*}
U(a)\mathcal{A}_{H_{j'}}U(a)^* = \mathcal{A}_{H_{j'}}.
\end{equation*}
Indeed, this follows from $U$ being a projective representation of an abelian group, since for any $b \in H_{j'}$,
\begin{equation*}
U(a)U(b)U(a)^* \propto U(a + b - a) = U(b),
\end{equation*}
where $\propto$ denotes proportionality by a complex scalar. Therefore, $U(a)$ must take $\mathcal{B}_\ell$'s basis into one which is also unbiased with respect to the original set of MUBs. However, the first theorem tells us that our completion algebras are unique in this respect, so $U(a)\mathcal{B}_\ell U(a)^*$ must be some $\mathcal{B}_{\ell'}$. Now, let $Q \in \mathcal{B}_\ell$ be a rank one orthogonal projection, corresponding to a vector of $\mathcal{B}_\ell$'s basis. We have that $QU(a)^*Q$ is a scalar multiple of $Q$, so
\begin{equation*}
\Tr((U(a)QU(a)^*)Q) = \Tr(U(a)(QU(a)^*Q)) \propto \Tr(U(a)Q) = 0.
\end{equation*}
The final equality holds because $U(a)$ is a traceless element of $\mathcal{A}_{H_j} \neq \mathcal{B}_\ell$. Thus, $U(a)QU(a)^*$ is orthogonal to $Q$. Since the bases corresponding to $\mathcal{B}_{\ell'}$ for $\ell' \neq \ell$ are unbiased with respect to that of $\mathcal{B}_\ell$, $U(a)QU(a)^*$ must lie in $\mathcal{B}_\ell$. By our previous observation, this means that $U(a)\mathcal{B}_\ell U(a)^* = \mathcal{B}_\ell$, so $a$ determines a permutation of the rank one orthogonal projections of $\mathcal{B}_\ell$ without fixed points, as desired. Lastly, it is simple to check that map composition corresponds to group multiplication in $G$.

Now, we have a well-defined action, and the lack of fixed points under $a \neq 0$ means that all stabilizers under $H_j$ are trivial. By the orbit-stabilizer theorem, each rank one projection $Q \in \mathcal{B}_\ell$ must then have full orbit size $|H_j| = d = \dim(\mathcal{B}_\ell)$, so the action is transitive.
\end{proof}

With this lemma in hand, our second theorem is straightforward.

\begin{theorem}\label{thm:niceRigidity}
Let $\mathcal{E}$ be a nice error basis for $M_d(\C)$ with abelian index group. If an $\mathcal{E}$-nice set of at least $d+1-\sqrt{d}$ MUBs can be completed to a full set of $d+1$ MUBs, then this completion is unique and $\mathcal{E}$-nice.
\end{theorem}

\begin{proof}
We have already demonstrated uniqueness; assume existence and label the unique subalgebras corresponding to the added MUBs by $\mathcal{B}_1, \dots, \mathcal{B}_k$, where $k \leq \sqrt{d}$. Suppose the nice error basis $\mathcal{E} = \{U(g) \mid g \in G\}$ has abelian index group $G$ with commutator map $\sigma$, and let $H_1, \dots, H_{d+1-k}$ be the subgroups of $G$ corresponding to the original $\mathcal{E}$-nice bases. Fix a non-zero element $a \in H_1$ and a rank one orthogonal projection $Q$ in some $\mathcal{B}_\ell$. By Lemma \ref{le1}, we know that for each $j = 2, \dots, d + 1 - k$, there exists a non-zero $b_j \in H_j$ such that
\begin{align*}
&U(a)QU(a)^* = U(b_j)QU(b_j)^*\\
\iff & U(b_j)^*U(a)Q = QU(b_j)^*U(a)\\
\iff & U(a - b_j)Q = QU(a - b_j).
\end{align*}
Now, for any other rank one orthoprojection $Q' \in \mathcal{B}_\ell$, take $g \in G$ such that $U(g)QU(g)^* = Q'$. It is easy to check that $U(r)U(s)^* = \sigma(s,r)U(s)^*U(r)$ for $r,s \in G$, so conjugating the final equation by $U(g)$ gives
\begin{align*}
    &U(g)U(a\!-\!b_j)QU(g)^* = U(g)QU(a\!-\!b_j)U(g)^*\\
    \!\iff &\sigma(g, a\!-\!b_j)U(a\!-\!b_j)U(g)QU(g)^* \!= \sigma(g, 
    a\!-\!b_j)U(g)QU(g)^*U(a\!-\!b_j)\\
    \!\iff &U(a\!-\!b_j)Q' = Q'U(a\!-\!b_j),
\end{align*}
i.e.\ $U(a - b_j)$ commutes with $Q'$ as well. But to commute with all rank one projections of $\mathcal{B}_\ell$, we must have $U(a - b_j) \in \mathcal{B}_\ell$, because $\mathcal{B}_\ell$ is maximally abelian. Next, define the generated subgroup
\begin{equation*}
N_\ell = \langle a - b_2, \dots, a - b_{d+1-k} \rangle,
\end{equation*}
noting that $|N_\ell| \geq d +1 - k$ since the $d-k$ generators are all distinct and non-zero. Because $\mathcal{B}_\ell$ is a subalgebra, it follows that $\mathcal{A}_{N_\ell}$ is contained in $\mathcal{B}_\ell$ and $|N_\ell| \leq \dim\mathcal{B}_\ell = d$. Additionally, as a subgroup of $G$, $N_\ell$ must have order dividing $|G| = d^2$. With $k \leq \sqrt{d}$, Proposition \ref{prop:divisors} (see Appendix) implies that the only divisor of $d^2$ between $d - (k - 1)$ and $d$ is $d$ itself, so we have that $|N_\ell| = d$ and $\mathcal{B}_\ell = \mathcal{A}_{N_\ell}$. Since $\ell$ was arbitrary, it follows that all of the added MUBs are generated by subgroups of $G$; i.e.\ the complete set of MUBs is also $\mathcal{E}$-nice.
\end{proof}

\subsection{Uncompletability results}
We have now developed the machinery to establish several uncompletability results for MUBs. The first of these is a consequence of the following connection between our groups of interest and combinatorial $k$-nets \cite{aschbacher07}.

\begin{lemma}\label{lem:groupKnets}
Let $G$ be a group of order $d^2$ together with a collection $\mathcal{C}$ of subgroups of $G$ of order $d$ with trivial pairwise intersections. Then the incidence structure whose points are elements of $G$ and whose lines are the left cosets of the subgroups defines a combinatorial $|\mathcal{C}|$-net.
\end{lemma}

\begin{proof}
Clearly, the disjoint cosets of a fixed subgroup form a parallel class. If two cosets of different subgroups did not intersect at a single element, then we would certainly have an intersection of two cosets $Ax$ and $By$ with size greater than 1. However, the non-empty intersection $Ax \cap By$ is known to be a coset of the trivial intersection $A \cap B$.
\end{proof}

With this result in hand, we can strengthen Theorem \ref{thm:niceRigidity}, proving an analog of Corollary \ref{cor:completionUniqueness} for nice MUBs with abelian index groups.

\begin{corollary}\label{cor:groupRigidity}
Let $\mathcal{E}$ be a nice error basis for $M_d(\C)$ with abelian index group $G$, and suppose that $\mathcal{C}$ is a collection of at least $d+1-\sqrt{d}$ subgroups of $G$ corresponding to an $\mathcal{E}$-nice set of MUBs. If these can be completed to a full set of $d+1$ MUBs, then this completion is unique and $\mathcal{E}$-nice. Furthermore, if a subgroup $H \leq G$ intersects each subgroup of $\mathcal{C}$ trivially, then $H$ must generate one of the completion bases.
\end{corollary}

\begin{proof}
By Lemma \ref{lem:groupKnets}, $\mathcal{C}$ corresponds to a combinatorial $(d+1-k)$-net $\mathcal{N}$ of order $d$, for some $k \leq \sqrt{d}$. If the set of MUBs can be completed, then Theorem \ref{thm:niceRigidity} states that this completion is unique and $\mathcal{E}$-nice, corresponding to a set of subgroups that provide a completion of $\mathcal{N}$ to an affine plane. By Corollary \ref{cor:completionUniqueness}, this affine plane is unique and contains any parallel class (or line, for that matter) extending $\mathcal{N}$. Since $H$ intersects each subgroup of $\mathcal{C}$ trivially, it corresponds to a parallel class extending $\mathcal{N}$ and must therefore be one of the completion subgroups.
\end{proof}

Next, we prove our first uncompletability result.

\begin{corollary}\label{cor:uncompletability1}
Let $\mathcal{E}$ be a nice error basis for $M_d(\C)$ with abelian index group $G$ and commutator map $\sigma$. Suppose that $\mathcal{C}$ is a collection of at least $d+1-\sqrt{d}$ subgroups of $G$ corresponding to an $\mathcal{E}$-nice set of MUBs. If there exists another subgroup $H \leq G$ which intersects those of $\mathcal{C}$ trivially, such that $\sigma|_{H \times H} \not\equiv 1$, then the set of MUBs cannot be completed.
\end{corollary}

\begin{proof}
If the set of MUBs could be completed, then the previous corollary would imply that $H$ generates one of the completion bases. However, since $\sigma|_{H \times H} \not\equiv 1$, the the subalgebra $\mathcal{A}_H$ is non-abelian and does not correspond to a basis.
\end{proof}

Now, we can provide an explicit example of an uncompletable set of MUBs. 

\begin{example}
Recall the subgroups $R_{x,y}$ and $S$ from Example \ref{ex:discreteWeylPrimeSquare}. If $p>2$ is an odd prime, then the system of $p^2 - p + 1$ nice MUBs in $\C^{p^2}$ corresponding to $S$ and $R_{x,y}$ for $x,y \in \Z_p, x \neq 0,$ cannot be completed.
\end{example}

\begin{proof}
The MUBs are generated by
\begin{align*}
    R_{x,y} &= \inner{(1,x,y,0), (0,-y,-Dx,1)}, \quad x,y \in \Z_p, x \neq 0,\\
    S &= \inner{(1,0,0,0),(0,0,1,0)}.
\end{align*}
However, the subgroup
\begin{equation*}
    T = \inner{(1,0,0,1),(0,1,-1,0)} \leq \Z_p^4
\end{equation*}
has order $p^2$, intersects these trivially, and does not correspond to an abelian subalgebra. Indeed, since $p > 2$, $\sigma((1,0,0,1),(0,1,-1,0)) = \omega^2 \neq 1$, where $\omega = e^{2\pi i / p}$. Therefore, the set of MUBs cannot be completed.
\end{proof}

In \cite{szanto16}, Sz\'{a}nto constructs a similar system of $p^2 - p + 2$ MUBs that cannot be completed. Specifically, he examines the subgroups of $\Z_p^4$ given by
\begin{align*}
    A_{x,y} &= \inner{(0,1,x,y^{-1}(1-Dx^2)),(1,0,-y,Dx)},\\
    B_x &= \inner{(0,1,x,0), (1,0,0,-xD)},
\end{align*}
where $D$ is again a quadratic nonresidue. Taking $q$ such that $q^2 = -D^{-1}$ and prime $p \equiv 3 \pmod{4}$, he shows that the MUBs of $\C^{p^2}$ associated with $B_q$, $B_{p-q}$, and $A_{x,y}$ for \mbox{$x,y \in \Z_p$}, \mbox{$y \neq 0$}, cannot be completed. Furthermore, he notes that these subgroups have trivial intersections with the remaining $B_x$ subgroups, which fail to generate abelian subalgebras. Thus, Corollary \ref{cor:uncompletability1} implies that this uncompletability result still holds if one of the MUBs is discarded.

While Sz\'{a}nto's proof requires an extra basis and is only valid for dimension $p \equiv 3 \pmod{4}$, he actually proves something stronger, that this construction cannot be extended by even a single basis vector. Previously, many systems of MUBs have been shown to satisfy a weaker form of unextendibility \cite{thas16}.

\begin{definition}
Let $\mathcal{E}$ be a nice error basis for $M_d(\C)$. A set of $\mathcal{E}$-nice MUBs in $\C^d$ is called \textit{weakly unextendible} if there does not exist another mutually unbiased $\mathcal{E}$-nice basis.
\end{definition}

The usual definition is more general but reduces to ours in the case of nice MUBs, which appear most often in applications. If a set of $\mathcal{E}$-nice MUBs is large enough, then Theorem \ref{thm:niceRigidity} tells us that any basis extending it must also be $\mathcal{E}$-nice, so we have the following.

\begin{corollary}
A weakly unextendible set of at least $d+1-\sqrt{d}$ nice MUBs in $\C^d$ cannot be completed.
\end{corollary}

This simple corollary allows us to translate many previous weak unextendibility results into uncompletability results. Letting $\mathcal{E}_p = \{X_p^{j}Z_p^{\ell} \otimes X_p^{r}Z_p^{s} \mid (j,\ell,r,s) \in \Z_p^4\}$ denote the nice error basis introduced in Example \ref{ex:discreteWeylPrimeSquare}, we can express two theorems of Thas \cite{thas16} and Mandayam, Bandyopadhyay, Grassl, and Wootters \cite{mandayam14} in terms of nice MUBs.

\begin{theorem}[Mandayam et al. \cite{mandayam14}]\label{thm:weaklyUnextendible}
Given subgroups $H_1, \dots, H_5$ of $\Z_2^4$ corresponding to a complete set of $\mathcal{E}_2$-nice MUBs in $\C^4$, there exists exactly one subgroup $H \leq Z_2^4$ in $H_1 \cup H_2 \cup H_3$ which generates a distinct basis. This subgroup $H$, together with the remaining two subgroups $H_4$ and $H_5$, generates a weakly unextendible set of 3 $\mathcal{E}_2$-nice MUBs.
\end{theorem}

\begin{theorem}[Thas \cite{thas16}]
For each prime $p$, there exists an unextendible set of $p^2 - p + 1$ or $p^2 - p + 2$ $\mathcal{E}_p$-nice MUBs in $\C^{p^2}$.
\end{theorem}

Since these systems are sufficiently large, we have the following result.

\begin{proposition}
These MUBs cannot be completed.
\end{proposition}

We conclude this article by noting that in addition to the uniqueness result, in \cite{bruck} Bruck also derives a condition for \textit{existence}. If a classical $k$-net misses only a few parallel classes for completeness, then -- if it can be completed at all -- its completion is unique. However, it may happen that it cannot be completed. But Bruck proves that if even fewer parallel classes are missing, then the existence of a completion is automatic. Having examined uniqueness, we wonder whether one could derive such an existence result in our general setting and hence obtain it for MUBs as well. We believe that this could be an interesting topic for future investigations.

\textbf{Acknowledgement.} The authors are grateful for the environment provided by Budapest Semesters in Mathematics and for fruitful discussions on finite geometry with Zsuzsa Weiner.

\vfill

\pagebreak

\appendix
\section*{Appendix}\label{sec:appendix}

Theorem \ref{thm:rigidity} requires two inequalities pertaining to the eigenvalue bounds
\begin{equation*}
    \lambda_\pm = \frac{1}{2d}\Bigg(d - 2 \pm \sqrt{\big(d-2\big)^2 - 4\big(d-1\big)\big(k - 3 + \tfrac{1}{k}\big)}\Bigg).
\end{equation*}

\begin{lemma}\label{lem:inequality1}
For $2 \leq k \leq \sqrt{d}$,
\begin{equation*}
    \frac{\sqrt{(d-1)(k-2)}}{d-2}  < \frac{\frac{d}{d-1}\lambda_+}{1+\frac{d}{d-1}\lambda_+}.
\end{equation*}
\end{lemma}
\begin{proof}
Since $\lambda_+$ is monotonically decreasing in $k$ and $t \mapsto t/(1+t)$ is monotonically increasing in $t$, the RHS is monotonically decreasing in $k$. Clearly, the LHS is monotonically increasing in $k$. Thus, it suffices to prove the statement for the extreme case of $d = k^2$. For $k=2,3$, we simply verify by substitution. For larger $k$, we strengthen the inequality slightly, using $\sqrt{d-1} < \sqrt{d}=k$, and rearrange to obtain
\begin{equation}\label{eq:kdInequality1}
    k\sqrt{k-2} < \frac{d}{d-1}(d-2-k\sqrt{k-2})\lambda_+.
\end{equation}
Next, we bound $\lambda_+$ from below by
\begin{align*}
    \lambda_+ &= \frac{1}{2d}\Bigg(d - 2 + \sqrt{\big(d-2\big)^2 - 4\big(d-1\big)\big(k - 3 + \tfrac{1}{k}\big)}\Bigg)\\
    &\geq \frac{1}{2d}\Bigg(d - 2 + \frac{1}{d}\Big(\big(d-2\big)^2 - 4\big(d-1\big)\big(k - 3 + \tfrac{1}{k}\big)\Big)\Bigg),
\end{align*}
observing that the argument of the square root is always less than $d^2$. We now substitute this bound for $\lambda_+$ in 
(\ref{eq:kdInequality1}) and rearrange to obtain that our claim surely holds if the expression
\begin{equation*}
    \frac{4}{k^3}+ \frac{-8+2\sqrt{k-2}}{k^2} + \frac{2-4\sqrt{k-2}}{k} +2(1+\sqrt{k-2})- 2(1+\sqrt{k-2})k+k^2
\end{equation*}
is positive. Omitting positive terms $4/k^3$, $2\sqrt{k-2}/k^2$, $(-8/k^2 + 2/k)$ and using that $-4\sqrt{k-2}/k \geq -\sqrt{k-2}$ for $k \geq 4$, it is then 
%\begin{equation*}
%      k^2 - 2(1+\sqrt{k-2})k + 2 + \sqrt{k-2}>0
%\end{equation*}
an elementary exercise to justify that the above expression is indeed positive.
\end{proof}

\begin{lemma}\label{lem:inequality2}
For $2 \leq k \leq \sqrt{d}$,
\begin{equation*}
    \frac{\sqrt{k-2}(k-1)}{\sqrt{k}} < \frac{d-2-d\lambda_-}{\sqrt{d-1}}.
\end{equation*}
\end{lemma}
\begin{proof}
Since $\lambda_-$ is monotonically increasing in $k$, the RHS is monotonically decreasing in $k$, and the LHS is clearly monotonically increasing in $k$. Hence, it suffices to prove the statement for the extreme case of $d = k^2$. We first observe that the inequality is trivial if $k=2$; otherwise we strengthen it by noting $\sqrt{(k-2)/k} \leq 1 - 1/k$ and $\sqrt{d-1} \leq \sqrt{d}=k$, and rearrange to obtain
\begin{equation}\label{eq:kdInequality2}
    (k-1)^2 \geq d-2-d\lambda_-.
\end{equation}
Next, we bound $\lambda_-$ from above by
\begin{align*}
    \lambda_- &= \frac{1}{2d}\Bigg(d - 2 - \sqrt{\big(d-2\big)^2 - 4\big(d-1\big)\big(k - 3 + \tfrac{1}{k}\big)}\Bigg)\\
    &\leq \frac{1}{2d}\Bigg(d - 2 - \frac{1}{d}\Big(\big(d-2\big)^2 - 4\big(d-1\big)\big(k - 3 + \tfrac{1}{k}\big)\Big)\Bigg),
\end{align*}
observing that the argument of the square root is always less than $d^2$. Substituting this bound for $\lambda_-$ in (\ref{eq:kdInequality2}) and simplifying, we strengthen the inequality a final time to
\begin{equation*}
\frac{2}{k^3}(k^3+2k-1) > 0,
\end{equation*}
which is easily verified for $k > 1$.
\end{proof}

Theorem \ref{thm:niceRigidity} requires an elementary result from number theory.
\begin{proposition}\label{prop:divisors}
The gap between $d$ and the next smallest divisor of $d^2$ is at least $\lfloor\sqrt{d}\rfloor$.
\end{proposition}
\begin{proof}
Suppose that $d - j$ divides $d^2$, for some positive integer $j$. Then,
\begin{equation*}
    \frac{d^2}{d-j} = d + j + \frac{j^2}{d-j}
\end{equation*}
is an integer, forcing its last term $j^2/(d-j) =: f(j)$ to be one as well. Using the monotonicity of $f$, it is easy to check that $0 < f(j) < 1$ for $1 \leq j \leq \sqrt{d}-1$. Thus, the integrality of $f(j)$ requires $j > \sqrt{d}-1$, that is $j \geq \lfloor \sqrt{d} \rfloor$.
\end{proof}

\printbibliography

\end{document}